\documentclass[a4paper,twocolumn,11pt,accepted=2020-11-30]{quantumarticle}
\pdfoutput=1
\usepackage{amsmath,amsthm}
\usepackage{amssymb}
\usepackage[numbers,sort&compress]{natbib}

\newcommand{\la}{\lambda}

\theoremstyle{plain}
\newtheorem{prop}{Proposition}
\newtheorem{lemma}[prop]{Lemma}
\newtheorem{corr}[prop]{Corollary}

\newtheorem{thm}[prop]{Theorem}

\theoremstyle{definition}

\newtheorem{rem}[prop]{Remark}

\title{Synthesis of CNOT-Dihedral circuits with optimal number of two qubit gates}
\author{Shelly Garion}
\address{IBM Quantum, IBM Research Haifa, Haifa University Campus, Mount Carmel, 
	Haifa, 3498825, Israel}
\email{shelly@il.ibm.com}

\author{Andrew W. Cross}
\address{IBM Quantum, IBM T.J. Watson Research Center, Yorktown Heights, NY 10598, USA}
\email{awcross@us.ibm.com}

\begin{document}

\begin{abstract}
In this note we present explicit canonical forms for all the elements in the two-qubit 
CNOT-Dihedral group, 
with minimal numbers of controlled-$S$ ($CS$) and controlled-$X$ ($CX$) gates, 
using the generating set of quantum gates $[X, T, CX, CS]$.
We provide an algorithm to successively construct the $n$-qubit CNOT-Dihedral group,
asserting an optimal number of controlled-$X$ ($CX$) gates.
These results are needed to estimate gate errors via non-Clifford randomized benchmarking
and may have further applications to circuit optimization over fault-tolerant gate sets.
\end{abstract}

\maketitle

\section{Introduction}
\emph{Randomized Benchmarking (RB)}~\cite{RB08,RB10,RB11} is a well-known algorithm
that provides an efficient and reliable experimental estimation of an average error-rate 
for a set of quantum gate operations, by running sequences of random gates from the 
\emph{Clifford} group that should return the qubits to the initial state. 
RB techniques are scalable to many qubits since the Clifford group can be efficiently 
simulated (in polynomial time) using a classical computer~\cite{AG,BM2,GK,MR}.
RB can also be used to characterize specific interleaved gate errors~\cite{RBint},
coherence errors~\cite{RBPurity1, RBPurity2} and leakage errors~\cite{RBleak}.
RB methods were generalized to certain single qubit non-Clifford gates,
like the $T$-gate~\cite{RBDihedral}. 
In~\cite{NonCliff} the authors presented a scalable RB procedure to benchmark 
important non-Clifford gates, such as the controlled-$S$ gate and controlled-controlled-$Z$ gate,
which belong to a certain group called the \emph{CNOT-Dihedral} group.

Certain CNOT-Dihedral groups have two key characteristics in common with the Clifford group. 
First, these groups have elements with concise representations that can be efficiently manipulated
~\cite{ACR, NonCliff}. 
Second, these groups are the set of transversal (fault-tolerant) gates for certain quantum 
error-correcting codes~\cite{Bo, BM, BK, GC, JKY, ZCC}.
Since the Clifford gates together with the $T$ gate form a universal set of gates, 
there are many papers aiming to optimize the number of $T$ gates~\cite{CH,GKMR,HC,NRSCM,RS}. 
Additional methods aim to minimize the count of controlled-$X$ ($CX$) gates in universal circuits~\cite{YSYI}, and in particular, in controlled-$X$-phase circuits~\cite{AAM, MSM}.

In addition, as the Clifford gate together with the controlled-$S$ ($CS$) gate also forms a 
universal set of gates, an algorithm has recently been introduced to construct a circuit with an optimal number of $CS$ gates given a two-qubit Clifford+$CS$ operator~\cite{GRT}.
Another example is the controlled-controlled-$Z$ gate, which is equivalent to the Toffoli gate 
(up to single qubit gates), that can be decomposed into 6 $CX$ gates and single qubit gates, 
but requires only 5 two-qubit gates in its decomposition if the $CS$ and $CS^{-1}$
gates are also available~\cite{Toffoli}.

It is therefore important to efficiently present the elements in the CNOT-Dihedral group
using a minimal number of physical basic gates, in particular, two-qubit gates like the 
controlled-$X$ ($CX$) and controlled-$S$ ($CS$) gates.

\medskip

Recall that $X$ is the Pauli gate defined as 
$$X=\begin{pmatrix} 0 & 1 \\1 & 0 \end{pmatrix}$$
Fix an integer $m$ and define
$$T(m) =\begin{pmatrix} 1 & 0 \\0 & e^{2\pi i / m} \end{pmatrix}$$
By abuse of notation we will denote $T=T(m)$, although the $T$ gate is usually defined as 
$T(8) = \begin{pmatrix}1 & 0 \\ 0 & e^{2\pi i / 8}\end{pmatrix}$.

The single-qubit \emph{Dihedral} group is generated by the $X$ and $T=T(m)$ gates
(up to a global phase) and contains $2m$ elements,
\begin{equation}\label{def:dihedral}
\begin{split}
&\langle X, T \rangle / \langle \la I: \la \in \mathbb{C} \rangle = \\
&\{X^l T^k: l \in \{0,1\}, k \in \{0,\dots,m-1\}\}.
\end{split}
\end{equation}

More generally, the \emph{CNOT-Dihedral} group on $n$ qubits $G=G(m)$ 
is generated by the gates $X$, $T=T(m)$ and controlled-$X$ ($CX$), up to a global phase 
(see~\cite{NonCliff} for details),
\begin{equation}\label{def:cnotdihedral}
\begin{split}
G = G(m) =& \langle X_i, T_i, CX_{i,j}: \\
&i,j \in \{0,\dots,n-1\} \rangle / \\
&\langle \la I: \la \in \mathbb{C} \rangle,
\end{split}
\end{equation}
where the controlled-$X$ ($CX$) gate is defined as
$$CX = \begin{pmatrix} 1 & 0 & 0 & 0 \\ 0 & 1 & 0 & 0 \\ 
0 & 0 & 0 & 1 \\ 0 & 0 & 1 & 0 \end{pmatrix}$$

When $m$ is not a power of $2$, the group $G=G(m)$ has double exponential order as a function of the number of qubits $n$. 
In the special case when $m$ is a power of two, the group is only exponentially large and we can represent 
its elements efficiently (see~\cite{NonCliff}).
Elements of $G(m)$ belong to level $\log_2 m$ of the Clifford hierarchy when $m$ is a power of two~\cite{GC, JKY} 
and this is related to the fact that they are the transversal gates of certain $m$-dimensional quantum codes~\cite{Bo}.

Again, by abuse of notation we denote 
$S=T^2=T(m)^2 = \begin{pmatrix}1 & 0 \\ 0 & e^{4\pi i/m}\end{pmatrix}$, 
although the $S$ gate is usually defined as 
$T(8)^2=T(4) = \begin{pmatrix}1 & 0 \\ 0 & i\end{pmatrix}$. 
Observe that $S$ has order $m/2$ if $m$ is even, and order $m$ if $m$ is odd,
namely, $S$ has order $m/d$ where $d=\gcd(m,2)$.

The controlled-$S$ ($CS$) gate belongs to $G$ and can be written as
\begin{equation}\label{def:CS}
\begin{split}
CS_{i,j} &= T_iT_j \cdot CX_{i,j} \cdot I_iT^\dagger_j \cdot CX_{i,j} \\ 
&= \begin{pmatrix} 1 & 0 & 0 & 0 \\ 0 & 1 & 0 & 0 \\ 
0 & 0 & 1 & 0 \\ 0 & 0 & 0 & e^{4\pi i / m} \end{pmatrix},
\end{split}
\end{equation}
where $T_iT_j$ means the tensor product $T_i \otimes T_j$.
In the case where $m=8$, the $CS$ gate is less expensive to physically implement than one $CX$ gate\footnote{Up to single-qubit rotations, the gate is equivalent to controlled-$\sqrt{X}$ gate, so it can be implemented by evolving for half the duration of a controlled-$X$ gate~\cite{CSGate}.}
which makes it an alternative to $CX$ for improving circuit decompositions.

We focus on the case where $n=2$. The following two Theorems provide 
canonical forms for all the elements in the two-qubit CNOT-Dihedral group, 
such that the numbers of $CS$ and $CX$ gates are optimal.
This is analogous to the description in~\cite{RB2} of the elements in the 
two-qubit Clifford group.

\begin{thm}\label{thm:CS}
Consider the $CS$-Dihedral subgroup on two qubits, namely the two-qubit group
generated by the gates $X$, $T=T(m)$ and $CS$ (controlled-$S$), where $S=T^2$, and denote $d=\gcd(m,2)$. 
Then this group has $\frac{4m^3}{d} = \frac{m}{d}(2m)^2$ elements of the following form:
$$U = CS_{0,1}^e \cdot X_0^k X_1^{k'} \cdot T_0^l T_1^{l'}$$
where $k,k' \in \{0,1\}$, $l,l' \in \{0,\dots,m-1\}$, 
$e \in \{0,1,\dots,m/d-1\}=\{0, \pm 1, \pm 2, \dots, \pm \lceil \frac{m-d}{2d} \rceil \}$.
\end{thm}

\begin{thm}\label{thm:CSCX}
Let $G$ be the two-qubit CNOT-Dihedral group generated by the gates $X$, $T=T(m)$,
$CX$ and $CS$, where $S=T^2$, and denote $d=\gcd(m,2)$. 
Then this group has $24 \cdot m^3 / d$ elements, divided into the following four classes.
\begin{enumerate}
\item The first class is the {\bf CS-Dihedral subgroup} described in
Theorem~\ref{thm:CS} and has $\frac{4m^3}{d}$ elements, that can be written 
with no $CX$ gates.

\item The second class, called the {\bf CX-like class}, consists of 
$\frac{8m^3}{d} = 2 \cdot \frac{m}{d} \cdot (2m)^2$ elements, 
and contains all the elements of the following form,
which require exactly one $CX$ gate.
$$U = X_0^k X_1^{k'} \cdot T_0^{l} T_1^{l'} \cdot CX_{i,j} \cdot I_iT_j^{e}$$

\item The third class, called the {\bf Double-CX-like class}, consists of 
$\frac{8m^3}{d} = 2 \cdot \frac{m}{d} \cdot (2m)^2$ elements, 
and contains all the elements of the following form,
which require exactly two $CX$ gates.
$$
U = X_0^k X_1^{k'} \cdot T_0^l T_1^{l'} \cdot CX_{i,j} \cdot CX_{j,i} \cdot  I_iT_j^{e} 
$$	

\item The fourth class, called the {\bf Triple-CX-like class}, consists of 
$\frac{4m^3}{d} = \frac{m}{d} \cdot (2m)^2$ elements, 
and contains all the elements of the following form,
which require exactly three $CX$ gates.
$$U = X_0^k X_1^{k'} \cdot T_0^l T_1^{l'} \cdot CX_{0,1} \cdot CX_{1,0}
\cdot I_0T_1^{e} \cdot CX_{0,1} $$
\end{enumerate}	

where $k,k' \in \{0,1\}$, $l,l' \in \{0,\dots,m-1\}$, $e \in \{0,\dots,m/d-1\}$
and $(i,j) \in \{(0,1), (1,0)\}$.
\end{thm}

The following Theorem provides an algorithm to successively construct the 
$n$-qubit CNOT-Dihedral group. It is analogous to~\cite{Bra} that discusses 
the generation of the $n$-qubit Clifford group.
Case (1) of this Theorem shows that one can successively construct the CNOT-Dihedral group
asserting an optimal number of $CX$ gates, with a bound on the space to search 
these group elements (see Remark~\ref{rem}).
Moreover, one can also use the ``meet in the middle'' algorithm of~\cite{Meet} 
to synthesize gate sequences for the non-Clifford RB.

\begin{thm}\label{lem:main}
	Let $G=G(m)$ be the CNOT-Dihedral group on $n$ qubits, and denote $d=\gcd(m,2)$.
	
	\begin{enumerate}
		\item Let $F(r)$ be the subset of operators implementable by a circuit with $r$ $CX$ gates
		(and any number of $X$ and $T$ gates).
		Suppose $U$ is in $F(r+1)$, then 
		$$U = I_i T_j^l \cdot CX_{i,j} \cdot U'$$
		for some $U' \in F(r)$,  
		$i,j \in \{0,...,n-1\}$, $i \neq j$, $l \in \{0,\dots,m/d-1\}$.
		In particular, 
		$$|F(r+1)| \leq \frac{m(n^2-n)}{d}|F(r)|$$
		
		\item Let $H(r)$ be the subset of operators implementable by a circuit with $r$ $CS$ or 
		$CS^{\dagger}$ gates
		(and any number of $X$ and $T$ gates).
		Suppose $U$ is in $H(r+1)$, then 
		$$U = CS_{i,j}^e \cdot U'$$ for some 
		$U' \in H(r)$, $i,j \in \{0,...,n-1\}$, $i<j$, $e \in \{-1,1\}$.
		In particular, 
		$$|H(r+1)| \leq (n^2-n)|H(r)|$$
	\end{enumerate}
\end{thm}

\begin{rem}\label{rem}
We note that the bounds in Theorem~\ref{lem:main} are \emph{sharp} and cannot generally be improved, 
since there is an equality in certain cases.
Indeed, assume that $n=2$. If $H(r)$ is the subset of operators implementable by a circuit with $r$ $CS$ gates,
then $H(1)=2 \cdot H(0)$ (see Theorem~\ref{thm:CS}). 
If $F(r)$ is the subset of operators implementable by a circuit with $r$ $CX$ gates,
then $F(1)=\frac{2m}{d} \cdot F(0)$ (see Theorem~\ref{thm:CSCX}). 
\end{rem}

\begin{corr}
In order to generate all the elements in the $n$-qubit CNOT-Dihedral group $G=G(m)$ 
having at most $r$ $CX$ gates, 
the algorithm generates at most 
$$(2m)^{n} \cdot \left(\frac{m}{d}\right)^{r} \cdot (n^2-n)^{r}$$
group elements.
\end{corr}

\section{Useful identities and the proof of Theorem~\ref{lem:main}}
Consider quantum circuits on a fixed number of qubits $n$ that are products of controlled-X gates $CX$, 
bit-flip gates $X$, and single-qubit phase gates $T=T(m)$ satisfying $T|u\rangle:=e^{i\pi u/m}|u\rangle$. 
When these gates are applied to each qubit or pairs of qubits, they generate a group $G=G(m)$ of unitary operators 
that is an example of a CNOT-dihedral group. An element $U\in G$ acts on the standard basis as
\begin{equation}
U|x\rangle = e^{p(x)}|f(x)\rangle
\end{equation}
where $p(x)=p(x_1,\dots,x_n)$ is a polynomial called the \emph{phase polynomial} and $f(x)$ is an affine reversible function.
Since $x_j\in{\mathbb F}_2$, so $x_j^2=x_j$, the phase polynomial is
\begin{equation}
p(x)=\sum_{\alpha\subseteq\{0,1\}^n}p_\alpha x^\alpha
\end{equation}
where $x^\alpha=\prod_{j\in\alpha} x_j$. Furthermore, the coefficients can be chosen such that $p_\emptyset=0$ and $p_\alpha\in(-2)^{|\alpha|-1}{\mathbb Z}_{2m}$ otherwise (see~\cite{NonCliff}).

Recall the following useful identities in the Dihedral group defined in~(\ref{def:dihedral}) generated by the $T=T(m)$ and $X$ gates (up to a global phase),
\begin{equation}
\begin{split}\label{Dihedral}
T^{\dagger} =& T^{m-1} \\
XTX &= T^{\dagger} \\
TXT &= X \\
TXT^{\dagger} &= SX 
\end{split}
\end{equation}

We state here some useful identities in the CNOT-Dihedral group defined in~(\ref{def:cnotdihedral})
regarding the controlled-$S$ ($CS$) gate.
According to the definition of the $CS$ gate in~(\ref{def:CS}),
\begin{equation}\label{CS}
\begin{split}
CS_{i,j} &= T_iT_j \cdot CX_{i,j} \cdot I_iT^\dagger_j \cdot CX_{i,j} \\
&= CX_{i,j} \cdot I_{i}T^\dagger_{j} \cdot CX_{i,j} \cdot T_iT_j
\end{split}
\end{equation}

We deduce that
\begin{equation}\label{CSCX}
\begin{split}
CS_{i,j} \cdot CX_{i,j} &= T_iT_j \cdot CX_{i,j} \cdot I_iT^{\dagger}_j, \\
CX_{i,j} \cdot CS_{i,j} &= I_iT^{\dagger}_j \cdot CX_{i,j} \cdot T_iT_j
\end{split}
\end{equation}

Similarly,
\begin{equation}\label{CSdg}
\begin{split}
CS^\dagger_{i,j} &= 
T^\dagger_iT^\dagger_j \cdot CX_{i,j} \cdot I_iT_j \cdot  CX_{i,j} \\
&= CX_{i,j} \cdot I_i T_j \cdot  CX_{i,j} \cdot  T^\dagger_i T^\dagger_j
\end{split}
\end{equation}

We note that according to their definition, the $CS$ and $CS^{\dagger}$ gates (as well as their powers) 
are symmetrical, namely,
\begin{equation}\label{CSsym}
\begin{split}
CS_{j,i} = CS_{i,j} \\
CS_{j,i}^{\dagger} = CS_{i,j}^{\dagger}
\end{split}
\end{equation}

$T$ (and all its powers) commutes with the control and target of the $CS$ gate, 
namely,
\begin{equation}\label{CST}
\begin{split}
I_iT_j \cdot CS_{i,j} = CS_{i,j} \cdot I_iT_j, \\
T_iI_j \cdot CS_{i,j} = CS_{i,j} \cdot T_iI_j, \\
T_iT_j \cdot CS_{i,j} = CS_{i,j} \cdot T_iT_j
\end{split}
\end{equation}

In addition, we have the following relations between the $CS$ and $X$ gates,
\begin{equation}\label{CSXI}
\begin{split}
&X_i I_j \cdot CS_{i,j} \cdot X_i I_j = CS_{i,j}^{\dagger} \cdot I_iS_j 
= I_iS_j \cdot CS_{i,j}^{\dagger} \\
&I_i X_j \cdot CS_{i,j} \cdot I_i X_j = CS_{i,j}^{\dagger} \cdot S_iI_j = 
S_iI_j \cdot CS_{i,j}^{\dagger}\\
&X_i X_j \cdot CS_{i,j} \cdot X_i X_j = CS_{i,j} \cdot S_i^{\dagger}S_j^{\dagger} = S_i^{\dagger}S_j^{\dagger} \cdot CS_{i,j}
\end{split}
\end{equation}

We shall moreover use the following identities of the $CX$ gate.
$T$  (and all its powers) commutes with the control of $CX$, 
and $X$ (and all its powers) commutes with the target of $CX$, namely,
\begin{equation}\label{CXTX}
\begin{split}
I_iX_j \cdot CX_{i,j} &= CX_{i,j} \cdot I_iX_j, \\
T_iI_j \cdot CX_{i,j} &= CX_{i,j} \cdot T_iI_j
\end{split}
\end{equation}

In addition, we have the following relation between the control of $CX$ and 
the $X$ gate,
\begin{equation}\label{CXXX}
CX_{i,j} \cdot X_iI_j \cdot CX_{i,j} = X_iX_j
\end{equation}

Recall that the $Z$ gate is defined as $Z=\begin{pmatrix} 1&0 \\ 0&-1\end{pmatrix}$.
Then we have the following useful relation between the $CX$ gate and the $Z$ gate,
\begin{equation}\label{CXZ}
CX_{i,j} \cdot I_iZ_j \cdot CX_{i,j} = Z_iZ_j
\end{equation}

Finally, the product $CX_{i,j} \cdot CX_{j,i}$, which is in the iSWAP-like class of Clifford gates (see~\cite{RB2})), satisfies the following relation,
\begin{equation}\label{CXCXT} 
I_iT_j \cdot CX_{i,j} \cdot CX_{j,i} = 
CX_{i,j} \cdot CX_{j,i} \cdot T_i I_j 
\end{equation}

\medskip

Based on the above identities we can now prove Theorem~\ref{lem:main}.
\begin{proof}[Proof of Theorem~\ref{lem:main}]
	
1) There exists a product of single qubit gates $V=V_1 \dots V_n$, 
$V_k \in \langle X,T \rangle$ such that 
$U = V \cdot CX_{i,j} \cdot U'$ for some pair of qubits $i,j$.
Absorb $V_k$ for $k \notin \{i,j\}$ into $U'$, namely, 
$$U = X_i^k X_j^{k'} \cdot T_i^l T_j^{l'} \cdot  CX_{i,j} \cdot U'$$
for some $k,k',l,l'$ and $U' \in F(r)$.
Since $T_i^l$ commutes with the control of $CX_{i,j}$ by~(\ref{CXTX}), 
we can absorb $T_i^l$ in $U'$. Since $X_j^{k'}$ 
commutes with the target of $CX_{i,j}$ by~(\ref{CXTX}), 
we can also absorb $X_j^{k'}$ in $U'$.
Hence, $$U = X_i^k T_j^l \cdot CX_{i,j} \cdot U'$$ 
for some $k,l$ and $U' \in F(r)$.

If $k=1$ then according to~(\ref{CXXX}), $X_iI_j \cdot CX_{i,j} = CX_{i,j} \cdot X_iX_j$,
so we can replace $U$ by 
$$I_i T_j^l \cdot CX_{i,j} \cdot X_i X_j \cdot U' = I_i T_j^l \cdot CX_{i,j} \cdot U''$$ 
where $U'' \in F(r)$. We can therefore assume that $k=0$. 

If $m$ is even and $l \geq m/2$ then $T^{m/2} = Z$, so we can rewrite $U$ as 
$$U = I_i T_j^l \cdot I_i Z_j \cdot CX_{i,j} \cdot U'$$ 
for some $l<m/2$. According to~(\ref{CXZ}), 
$I_iZ_j \cdot CX_{i,j} = CX_{i,j} \cdot Z_iZ_j$, so we can replace $U$ by
$$I_i T_j^l \cdot CX_{i,j} \cdot Z_i Z_j \cdot U' = I_i T_j^l \cdot CX_{i,j} \cdot U''$$
where $U'' \in F(r)$. We can therefore assume that $l<m/2$ as needed.

\medskip

2) Similarly to (1) we can assume that 
$$U = X_i^k X_j^{k'} \cdot T_i^l T_j^{l'} \cdot  CS_{i,j}^e \cdot U'$$
for some $k,k',l,l',e=\pm 1$ and $U' \in H(r)$.
Since $T$ commutes with both control and target of $CS$ by~(\ref{CST}), 
we can absorb $T_i^l T_j^{l'}$ in $U'$ and so 
$$U = X_i^k X_j^{k'} \cdot CS_{i,j}^e \cdot U'$$
Now, by~(\ref{CSsym}) we may assume that $i<j$, 
and by~(\ref{CSXI}) we can absorb $X_i^k X_j^{k'}$ in $U'$
and assume that $U = CS_{i,j}^e \cdot U'$ for some $i<j$ and $e=\pm 1$ as needed.
\end{proof}

\section{The canonical forms and proofs of Theorems~\ref{thm:CS} and~\ref{thm:CSCX}}

From now on we will now assume that $G$ is the CNOT-Dihedral 
group on two qubits  $\{0,1\}$,
and describe canonical forms of the elements in $G$.
This is analogous to the description in~\cite{RB2} of the elements in the 
Clifford group on two qubits.

\begin{proof}[Proof of Theorem~\ref{thm:CS}]
	The proof follows by induction on the number $r$ of $CS$ and $CS^{\dagger}$ gates.
	Since $CS$ is of order $m/d$ then necessarily $r < \lceil \frac{m-d}{2d} \rceil$.
	
	Let $r=0$, then any $U \in H(0)$ can be written as
	$$U = X_0^k X_1^{k'} \cdot T_0^l T_1^{l'}$$
	where $k,k' \in \{0,1\}$, $l,l' \in \{0,\dots,m-1\}$,
	since such an element belongs to the direct product of the two single-qubit Dihedral groups.

	Let $r=1$, then according to Case (2) of Theorem~\ref{lem:main},
	any $U \in H(1)$ can be written as
	$$U = CS_{0,1}^e \cdot X_0^k X_1^{k'} \cdot T_0^l T_1^{l'}$$
	where $e \in \{1,-1\}$, $k,k' \in \{0,1\}$, $l,l' \in \{0,\dots,m-1\}$.
	
	Now assume that the Theorem holds for $H(r)$.
	According to Case (2) of Theorem~\ref{lem:main} and the induction assumption,
	any element $U \in H(r+1)$ can be written as 
	$$U = CS_{0,1}^{e} \cdot CS_{0,1}^{e'} \cdot U' = CS_{0,1}^{e+e'} \cdot U' $$
	where $U'\in \langle T,X \rangle$, $e=\pm 1$ and $e'=\pm r$, as needed.
	
	Note that all the elements obtained in this process are distinct, since an equality
	$CS_{0,1}^{e} \cdot U = CS_{0,1}^{e'} \cdot U'$ for some $e,e' \in \{0,\dots,m/d-1\}$ 
	and $U,U' \in \langle T,X \rangle$,
	implies that $CS_{0,1}^{e-e'} \in \langle T,X \rangle$, so necessarily $e=e'$ and $U=U'$.
\end{proof}

\begin{lemma}\label{lem:CS1CX1}
	Let $G$ be the CNOT-Dihedral group on two qubits.
	Then any element in $G$ which has exactly one $CS$ gate and one $CX$ gate
	can be rewritten as an element with no $CS$ gates and exactly
	one $CX$ gate.
\end{lemma}
\begin{proof} 
	According to Theorem~\ref{lem:main} we may assume w.l.o.g.\ that such an element $U$ 
	can be written as a product
	$$
	U = (U' \cdot CX_{0,1} \cdot I_0 T_1^l) \cdot
	(CS_{0,1}^{e} \cdot U'')
	$$
	where $U',U'' \in \langle T,X \rangle$, $l \in \{0,\dots,m/d-1\}$,
	$e \in \{1,-1\}$.
	
	Since $T$ commutes with the control and target of $CS$ by~(\ref{CST}), 
	we may absorb $T_1$ into $U''$, 
	and so $U$ can be rewritten as 
	\begin{equation*}
	\begin{split}
	U &= U' \cdot CX_{0,1} \cdot CS_{0,1}^e \cdot U''  \\ &=
	U' \cdot I_0 T_1^{-e} \cdot CX_{0,1} \cdot T_0^e T_1^e \cdot U''
	\end{split}
	\end{equation*}
	for some $U',U''$ by~(\ref{CSCX}).
	Therefore, $U = U' \cdot  CX_{0,1} \cdot U''$ for some $U',U''$, as needed.
\end{proof}

\begin{lemma}\label{lem:CX1}
	Let $G$ be the CNOT-Dihedral group on two qubits.
	Then any element in $G$ which has exactly one $CX$ gate and no $CS$ gates
	can be written either as:
	$$U =  X_0^k X_1^{k'} \cdot T_0^{l} T_1^{l'} \cdot CX_{0,1} \cdot I_0T_1^{l''}$$
	or: 
	$$U =  X_0^k X_1^{k'} \cdot T_0^{l} T_1^{l'} \cdot CX_{1,0} \cdot T_0^{l''}I_1$$
	where $k,k' \in \{0,1\}$, $l,l' \in \{0,\dots,m-1\}$ and $l'' \in \{0,\dots,m/d-1\}$.
	In particular, $G$ has $\frac{8m^3}{d} = 2 \cdot \frac{m}{d} \cdot (2m)^2$ such elements.
\end{lemma}
\begin{proof} 
The proof follows from Case (1) of Theorem~\ref{lem:main}.

Note that all the elements obtained in this process are indeed distinct.

First, an equality $U \cdot CX_{0,1} \cdot I_0 T_1^l = U' \cdot CX_{0,1} \cdot I_0 T_1^{l'}$
for some $U, U' \in \langle T,X \rangle$ and $l,l' \in \{0,\dots,m/d-1\}$, 
implies that $CX_{0,1} \cdot I_0 T_1^{l'-l} \cdot CX_{0,1} \in \langle T,X \rangle$,
hence either $l=l'$ and $U=U'$; 
or $m$ is even and $l-l'=m/2$, yielding a contradiction since $l,l'<m/2$.

Second, an equality $U \cdot CX_{0,1} \cdot I_0 T_1^l = U' \cdot CX_{1,0} \cdot T_0^{l'} I_1$
for some $U, U' \in \langle T,X \rangle$ and $l,l' \in \{0,\dots,m/d-1\}$, 
implies that $CX_{0,1} \cdot T_0^{-l'} T_1^{l} \cdot CX_{1,0} \in \langle T,X \rangle$,
yielding a contradiction.
\end{proof}

\begin{lemma}\label{lem:CX2}
	Let $G$ be the CNOT-Dihedral group on two qubits.
	Then any element in $G$ which has exactly two $CX$ gates and no $CS$ gates
	can be written either as:
	$$
	U = X_0^k X_1^{k'} \cdot T_0^l T_1^{l'} \cdot CX_{0,1} \cdot CX_{1,0} \cdot I_0T_1^{l''}
	$$
	or:
	$$
	U = X_0^k X_1^{k'} \cdot T_0^l T_1^{l'} \cdot CX_{1,0} \cdot CX_{0,1} \cdot T_0^{l''}I_1 
	$$
	where $k,k' \in \{0,1\}$, $l,l' \in \{0,\dots,m-1\}$ and $l'' \in \{0,\dots,m/d-1\}$.
	In particular, $G$ has $\frac{8m^3}{d} = 2 \cdot \frac{m}{d} \cdot (2m)^2$ such elements.
\end{lemma}
\begin{proof} 
According to Case (1) of Theorem~\ref{lem:main} and Lemma~\ref{lem:CX1} 
we may assume w.l.o.g. that such an element $U$ 
can be written as
$$ 
U = I_i T_j^l \cdot CX_{i,j} \cdot I_0 T_1^{l'} \cdot CX_{0,1} \cdot U'
$$	
where $U' \in \langle T,X \rangle$, $i,j \in \{0,1\}$, 
$l,l' \in \{0,...,m/d-1\}$. 
Hence, there are two options, either $(i,j)=(0,1)$ or $(1,0)$.

1) First, assume that $(i,j)=(0,1)$, then
$$ 
U = I_0 T_1^l \cdot CX_{0,1} \cdot I_0 T_1^{l'} \cdot CX_{0,1} \cdot U'
$$	

If $l'=0$ then $U \in \langle X,T \rangle$ and we are done.

Otherwise, according to~(\ref{CSdg}), 
$CX_{0,1} \cdot I_0T_1 \cdot CX_{0,1} = CS_{0,1}^{\dagger} \cdot T_0T_1$,
implying that 
\begin{equation*}
\begin{split}
CX_{0,1} \cdot I_0T_1^{l'} \cdot CX_{0,1} &= 
(CX_{0,1} \cdot I_0T_1 \cdot CX_{0,1})^{l'} \\
&= (CS_{0,1}^{\dagger} \cdot T_0T_1)^{l'} \\
&= CS_{0,1}^{-l'} \cdot T_0^{l'}T_1^{l'}
\end{split}
\end{equation*}

by~(\ref{CST}).
Thus we can write $U$ as an element in the subgroup generated by $CS$,
$X$ and $T$.

Then we are done by Theorem~\ref{thm:CS}.

\medskip

2) Now, assume that $(i,j)=(1,0)$, then we can write $U$ as
$$ 
U = T_0^lI_1 \cdot CX_{1,0} \cdot I_0 T_1^{l'} \cdot CX_{0,1} \cdot U'
$$	

By~(\ref{CXTX}), $T_1$ commutes with $CX_{1,0}$, so we may write $U$ as
$$ 
U = T_0^lT_1^{l'} \cdot CX_{1,0} \cdot CX_{0,1} \cdot U'
$$	

According to~(\ref{CXCXT}),
$T_0I_1 \cdot CX_{1,0} \cdot CX_{0,1} = 
CX_{1,0} \cdot CX_{0,1} \cdot I_0 T_1$, 
so we can absorb $T_0$ in $U'$.

Therefore,
$$ 
U = I_0T_1^{l'} \cdot CX_{1,0} \cdot CX_{0,1} \cdot U'
$$	

for some $U' \in \langle X, T \rangle$ and $l' \in \{0,...,m/d-1\}$ as needed.

\medskip
Similar argument as in the proof of Lemma~\ref{lem:CX1} shows that all the elements obtained in 
this process are indeed distinct.

First, an equality 
$U \cdot CX_{0,1} \cdot CX_{1,0} \cdot I_0 T_1^l = U' \cdot CX_{0,1} \cdot CX_{1,0} \cdot I_0 T_1^{l'}$
for some $U, U' \in \langle T,X \rangle$ and $l,l' \in \{0,\dots,m/d-1\}$, 
implies that $CX_{0,1} \cdot CX_{1,0} \cdot I_0 T_1^{l'-l} \cdot CX_{1,0} \cdot CX_{0,1} \in \langle T,X \rangle$,
implying that $l=l$ and $U=U'$.

Second, an equality 
$U \cdot CX_{0,1} \cdot CX_{1,0}\cdot I_0 T_1^l = U' \cdot CX_{1,0} \cdot CX_{0,1} \cdot T_0^{l'} I_1$
for some $U, U' \in \langle T,X \rangle$ and $l,l' \in \{0,\dots,m/d-1\}$, 
implies that 
$CX_{0,1} \cdot CX_{1,0} \cdot T_0^{-l'} T_1^{l} \cdot CX_{0,1} \cdot CX_{1,0}\in \langle T,X \rangle$,
yielding a contradiction.
\end{proof}

\begin{lemma}\label{lem:CX3}
	Let $G$ be the CNOT-Dihedral group on two qubits.
	Then any element in $G$ which has exactly three $CX$ gates and no $CS$ gates
	can be written as: 
	$$
	U = X_0^k X_1^{k'} \cdot T_0^l T_1^{l'} \cdot CX_{0,1} \cdot CX_{1,0} 
	\cdot I_0T_1^{l''} \cdot CX_{0,1}  
	$$
	where $k,k' \in \{0,1\}$, $l,l' \in \{0,\dots,m-1\}$ and $l'' \in \{0,\dots,m/d-1\}$.
	In particular, $G$ has $\frac{4m^3}{d} = \frac{m}{d} \cdot (2m)^2$ such elements.
\end{lemma}
\begin{proof} 
According to Case (1) of Theorem~\ref{lem:main} and Lemma~\ref{lem:CX2} 
we may assume w.l.o.g. that such an element $U$ can be written as 
$$
U =  I_i T_j^l \cdot CX_{i,j} \cdot I_0T_1^{l'} \cdot CX_{1,0} \cdot CX_{0,1} \cdot U' 
$$	

where $U' \in \langle T,X \rangle$, $i,j \in \{0,1\}$, 
$l,l' \in \{0,...,m/d-1\}$.

Hence, there are two options, either $(i,j)=(0,1)$ or $(1,0)$.

1) First, assume that $(i,j)=(1,0)$, then 
$$
U = T_0^l I_1 \cdot CX_{1,0} \cdot I_0T_1^{l'} \cdot CX_{1,0} \cdot CX_{0,1} \cdot U' 
$$

By~(\ref{CXTX}),  $T_1$ commutes with $CX_{1,0}$, so we can write $U$ as

\begin{equation*}
\begin{split}
U &= T_0^l I_1 \cdot I_0T_1^{l'} \cdot CX_{1,0} \cdot CX_{1,0} \cdot CX_{0,1} \cdot U' \\
&= T_0^l T_1^{l'} \cdot CX_{0,1} \cdot U'
\end{split}
\end{equation*}

Then we actually have only one $CX$ gate and we are done by Lemma~\ref{lem:CX1}.

\medskip
 
2) Now assume that $(i,j)=(0,1)$, then we can write $U$ as
$$
U =  I_0 T_1^l \cdot CX_{0,1} \cdot I_0T_1^{l'} \cdot CX_{1,0} \cdot CX_{0,1} \cdot U' 
$$

for some $U',U'',l,l'$.

By~(\ref{CXTX}), $T_1$ commutes with $CX_{1,0}$, so we can rewrite $U$ as
$$
U = I_0 T_1^l \cdot CX_{0,1} \cdot CX_{1,0} \cdot I_0T_1^{l'} \cdot CX_{0,1} \cdot U' 
$$	

According to~(\ref{CXCXT}),
$I_0T_1 \cdot CX_{0,1} \cdot CX_{1,0} = 
CX_{0,1} \cdot CX_{1,0} \cdot T_0 I_1$, therefore,
$$ 
U = CX_{0,1} \cdot CX_{1,0} \cdot T_0^{l} T_1^{l'} \cdot CX_{0,1} \cdot U'
$$	

for some $l,l' \in \{0,...,m/d-1\}$.

Now, by~(\ref{CXTX}), $T_0$ commutes with $CX_{0,1}$ and so we can absorb $T_0$ in $U'$,
thus
\begin{equation*}
\begin{split}
U &= CX_{0,1} \cdot CX_{1,0} \cdot I_0 T_1^{l'} \cdot CX_{0,1} \cdot U' \\
&= CX_{0,1} \cdot I_0 T_1^{l'} \cdot CX_{1,0} \cdot CX_{0,1} \cdot U'
\end{split}
\end{equation*}

by using~(\ref{CXTX}) again. 

\medskip

The same argument as in the proof of Lemma~\ref{lem:CX2} shows that all the elements obtained in 
this process are indeed distinct.
\end{proof}


\begin{proof}[Proof of Theorem~\ref{thm:CSCX}]
According to Corollary 1 in~\cite{NonCliff}, the CNOT-Dihedral group $G=G(m)$ 
on two qubits has exactly $24 \cdot m^3 /d$ elements. 

By Lemma~\ref{lem:CS1CX1}, there are no elements with both $CX$ and $CS$ gates.
The cases where there are only $CS$ gates were handled in Theorem~\ref{thm:CS}.
The remaining cases where there are only $CX$ gates were proved in 
Lemmas~\ref{lem:CX1}, ~\ref{lem:CX2} and~\ref{lem:CX3}.
\end{proof}


\bibliographystyle{plainnat}
\bibliography{bibliography}

\end{document}